\documentclass[11pt,reqno]{amsart}
\usepackage{amssymb}
\usepackage[usenames,dvipsnames]{color}
\usepackage{euscript}
\usepackage{multicol}
\usepackage{times} 
\usepackage{amsfonts,amssymb,amsbsy,amsmath,amsthm,dsfont}

\newcommand{\D}{\rm D}
\newcommand{\Q}{\mathcal{Q}}
\newcommand{\eps}{\varepsilon}

\newcommand{\F}{\mathcal{F}}

\newcommand{\N}{\mathbb{N}}

\newcommand{\R}{\mathbb{R}}

\newcommand{\id}{\mathds{1}}

\newtheorem{theorem}{Theorem}[section]
\newtheorem{corollary}[theorem]{Corollary}
\newtheorem{lemma}[theorem]{Lemma}
\newtheorem{proposition}[theorem]{Proposition}

\theoremstyle{definition}

\newtheorem{remark}[theorem]{Remark}
\newtheorem{assumption}[theorem]{Assumption}

\numberwithin{equation}{section}

\begin{document}

\title[Laplace operator with mixed boundary conditions]{On the lowest eigenvalue of Laplace operators with mixed boundary conditions}

\author {Hynek Kova\v{r}\'{\i}k}

\address {Hynek Kova\v{r}\'{\i}k, Dipartimento di Matematica \\ Universit\`a degli Studi di Brescia \\ Via Branze, \\38 - 25123, Bresica Italy}

\email {hynek.kovarik@ing.unibs.it}


\begin{abstract}
In this paper we consider a Robin-type Laplace operator on bounded domains. We study the dependence of its lowest eigenvalue on the boundary conditions and its asymptotic behaviour in shrinking and expanding domains. For convex domains we establish two-sided estimates on the lowest eigenvalues in terms of the inradius  and of the boundary conditions. 
\end{abstract}

\maketitle



\section{\bf Introduction}

Let $\Omega\subset\R^N$ be a bounded domain. Given a measurable function $\sigma :\partial\Omega\to \R$,
we consider the quadratic form
\begin{equation} \label{quad-f}
\int_\Omega |\nabla u|^2\, dx  + \int_{\partial\Omega} \sigma\, |u|^2\, d\nu,
\end{equation}
where $d\nu$ denotes the $(N-1
)$ dimensional surface measure on $\partial\Omega$. If $\Omega$ is regular enough and if $\sigma \in L^{N-1}(\partial\Omega)$, then by the boundary trace imbedding theorems,  equation \eqref{trace-imbed} below, it follows that the quadratic form \eqref{quad-f} is closed on $H^1(\Omega)$ and generates in $L^2(\Omega)$ a unique self-adjoint operator, the so-called Robin-Laplacian. The case $\sigma= +\infty $ then corresponds to the Dirichlet boundary conditions  while by choosing $\sigma=0$ we get the Neumann boundary conditions.

The lowest eigenvalue of the Robin-Laplacian, which we denote by $\lambda_1(\sigma, \Omega)$, is the main object of our interest. Problems related to the Robin-Laplacian have been intensively studied in the literature. Among other questions, various problems such as Faber-Krahn inequalities, Hardy inequalities, monotonicity properties of the lowest eigenvalue, and comparison between Robin and Dirichlet or Neumann eigenvalues were considered in the literature, see \cite{boss1, boss2, bg, cu, da, da2, gs1, kl, lp, pw, ph, sp, sp2}.

\smallskip

The purpose of the present paper is twofold. First we will study the dependence of the lowest eigenvalue of the Robin-Laplacian on the function $\sigma$. We start by addressing the following question: which functions $\sigma$ maximise (or minimise) $\lambda_1(\sigma, \Omega)$ among all positive functions from $L^1(\partial\Omega)$ with a fixed integral mean and with a support contained in a prescribed subset of the boundary? It turns out that while the minimising $\sigma$ generically does not exist, the maximising function exists and is unique. An explicit description is given in Theorem \ref{thm-sup}. Next one would like to know how big the resulting maximum is. Sharp two-sided estimates on the corresponding maximal eigenvalue are given in Propositions \ref{prop-ub}, \ref{prop-lb} and Corollary \ref{cor-lambdam}. 

In the second part of the paper we will study the properties of $\lambda_1(\sigma, \Omega)$ for a fixed $\sigma$. In Theorem \ref{prop:scaling} it will be shown that, contrary to the lowest eigenvalue of the Dirichlet-Laplacian,  $\lambda_1(\sigma, \Omega)$ scales in a different way when the domain $\Omega$ shrinks to zero respectively when $\Omega$ blows up to infinity. We then prove a two-sided bound for $\lambda_1(\sigma, \Omega)$ on convex domains with constant $\sigma$. In particular,  we will show that 
$$
\lambda_1(\sigma, \Omega) \, \asymp \, \frac{\sigma}{ R_\Omega(1+\sigma\, R_\Omega)}  \qquad \quad [\, \Omega\ \ \text{convex}, \ \sigma\  \ \text{constant} ],
$$
where $R_\Omega$ is the inradius of $\Omega$, see Theorem \ref{thm-convex}.

\section{\bf Preliminaries} \label{s-prelim}

\noindent 
Throughout the paper will always assume that the following condition is satisfied:

\begin{assumption} \label{ass-basic}
$\Omega$ is an open bounded and connected set with a boundary which satisfies the strong local Lipschitz condition, see e.g. \cite[Chap.4]{ad}.   
\end{assumption}

\noindent Under the above assumption a trace operator is well defined on $H^1(\Omega)$. More precisely, we have
\begin{equation} \label{trace-imbed}
\| u\|_{L^q(\partial\Omega)} \, \leq\, C \, \| u\|_{H^1(\Omega)}, \qquad \forall\ q \, \left\{
\begin{array}{l@{\qquad}l}
\leq \, 2(N-1)/(N-2)    &       {\rm if \ } \   N>2,  \\
< +\infty    &    {\rm if \ } \  N=2
\end{array}
\right. 
\end{equation}
with a compact imbedding. This follows from standard Sobolev imbedding theorems and trace inequalities, see e.g. \cite[Thm.5.22]{ad}. For a 
given $\sigma \in L^1(\partial\Omega)$ we then consider the functional $Q[\sigma, \cdot\, ]$ on $H^1(\Omega)$ defined by
\begin{equation} \label{form}
Q[\sigma,u ]= \frac{\int_\Omega |\nabla u(x)|^2 \,dx + \int_{\partial\Omega} \sigma(s)\, |u(s)|^2\, d\nu(s)}{\| u\|^2_{L^2(\Omega)} } 
\end{equation}
if the right hand side is finite and by $Q[\sigma, u] = + \infty$ otherwise.  Let 
\begin{equation} \label{lambda}
\lambda_1(\sigma, \Omega) = \inf_{u \in H^1(\Omega)} Q[\sigma, u]. 
\end{equation}

\begin{lemma} \label{lem-minimiser}
Let $\sigma \in  L^1(\partial\Omega)$ and assume that $\sigma\geq 0$. Then the functional $Q[\sigma, \cdot\, ]$ admits a positive minimiser $\psi\in H^1(\Omega)$ which satisfies
\begin{equation} \label{euler-lagr}
-\Delta  \psi = \lambda_1(\sigma, \Omega)\, \psi \quad \text{in \ } \Omega, \qquad \partial_n \psi + \sigma\, \psi =0 \quad \text{a.e. \ \ on \ } \partial\Omega,
\end{equation}
where $\partial_n$ denotes the outer normal derivative.
\end{lemma}

\begin{proof}
Let $\{u_j\}_{j\in\N}$ be a minimising sequence for $Q[\sigma, \cdot\, ]$. Without loss of generality we assume that $\|u_j\|_{L^2(\Omega)}=1$ for all $j\in\N$. Since $\{u_j\}$ is bounded in $H^1(\Omega)$, there exists a subsequence, which we still denote by $u_j$ and a function $\psi\in H^1(\Omega)$ such that $u_j \to \psi$ weakly in $H^1(\Omega)$. Next, from the compactness of the imbedding  $H^1(\Omega) \hookrightarrow L^2(\partial\Omega)$, see \eqref{trace-imbed}, it follows that the trace of $u_j$
converges strongly in $L^2(\partial\Omega)$ to the trace of $\psi$. Therefore we can find a subsequence $\{v_j\} \subset \{u_j\}$ such that 
$v_j |_{\partial\Omega} \to \psi|_{\partial\Omega}$ almost everywhere on $\partial\Omega$. By the weak lower semicontinuity of $\int_\Omega |\nabla u|^2$ and  the Fatou Lemma we thus obtain 
$$
\liminf_{j\to\infty} Q[\sigma, v_j] \geq Q[\sigma, \psi].
$$
Hence $\psi$ is a minimiser of $Q[\sigma, \cdot\, ]$. From the fact that $Q[\sigma, \psi] \geq Q[\sigma, |\psi| ]$ it follows that $\psi \geq 0$ in $\Omega$. Therefore, by Harnack inequality $\psi>0$ in $\Omega$. The Euler-Lagrange equation for $Q[\sigma, \cdot\ ]$ then gives \eqref{euler-lagr}.
\end{proof}

\begin{remark} The assumption $\sigma\geq 0$ in the above Lemma is necessary. Indeed, if $\sigma\in L^1(\partial\Omega)$ changes sign, then the functional $Q[\sigma, \cdot]$ might not even be bounded from below 
\end{remark}


\section{\bf Optimising problem for $\lambda_1(\sigma, \Omega)$} \label{s-optim}

\noindent 
In this section we will only assume that $0\leq \sigma\in L^1(\partial\Omega)$. Note that although this condition does not guarantee the finiteness of $Q[\sigma, u]$ for all $u\in H^1(\Omega)$, the quantity $\lambda_1(\sigma, \Omega)$ is well defined. 
Let $\Gamma\subseteq \partial\Omega$ be a closed subset of the boundary (which might coincide with $\partial\Omega$). For a given $m>0$ we define
$$
\Sigma_m := \Big\{ \sigma \in L^1(\partial\Omega)\, :\, \sigma \geq 0, \ \ \int_{\partial\Omega} \sigma = m, \ \ {\rm supp\, }\sigma \subseteq\Gamma \Big\}.
$$
Our goal is to study $\lambda_1(\sigma, \Omega)$ as a functional of $\sigma$ on $\Sigma_m$. To start with we show that the functional $\lambda_1(\cdot, \Omega)$ admits no minimum on $\Sigma_m$ when $N\geq 2$.

\subsection{The infimum} 

\begin{proposition} \label{prop-inf1}
Let $m>0$ and suppose that $N\geq 2$. Then $\lambda_1(\cdot, \Omega)$ has no minimiser on $\Sigma_m$ and $\inf_{\sigma\in \Sigma_m} \lambda_1(\sigma, \Omega)= 0$.
\end{proposition}

\begin{proof}
In the sequel we denote by $B(x,r)$ the open ball of radius $r$ centred in $x\in \R^N$. 
Let $s_0\in\partial\Omega$ and let $\sigma_n \geq 0$ be given by
$$
\sigma_n(s) = 
\Big\{
\begin{array}{l@{\qquad}l}
\alpha_n    &      {\rm if \ } \  s\in\  B(s_0, 2^{-n})\cap \partial\Omega ,  \\
0     &   {\rm elsewhere \ } ,
\end{array}
\Big. 
$$
where $\alpha_n$ is a positive constant chosen so that $\sigma_n \in\Sigma_m$ for all $n\in\N$.  Depending on the dimension we construct a family of test functions $u_n$  as follows:
$$
u_n(x) = \frac{-\log n}{\log(|x-s_0|)}  \qquad   \text{on \ } B(s_0, \frac 1n)\cap\Omega, \qquad u_n \equiv 1 \quad \text{on \ } \Omega\setminus B(s_0, \frac 1n), \quad N=2.
$$
$$
u_n(x) = n\, |x-s_0|   \qquad\qquad \text{on \ } B(s_0, \frac 1n)\cap\Omega, \qquad  u_n \equiv 1 \quad \text{on \ } \Omega\setminus B(s_0, \frac 1n), \quad N\geq 3.
$$
Then $u_n \in H^1(\Omega)$ for all $n\in\N$ and a direct calculation shows that 
$$
\lim_{n\to\infty} Q[\sigma_n, u_n] = 0. 
$$
This proves that $\inf_{\sigma\in \Sigma_m} \lambda_1(\sigma, \Omega) = 0$. 
To show that the infimum is not attained, assume that $\sigma\in \Sigma_m,\, m>0$. Then there exists an $\eps>0$ and $\gamma_\eps \subset \Gamma$ such that $\sigma \geq \eps$ on $\gamma_\eps$ and $\nu(\gamma_\eps)>0$. By the Poincar\'e inequality 
$$
\int_\Omega |\nabla u(x)|^2 \,dx + \eps \int_{\gamma_\eps} |u(s)|^2\, d\nu(s) \,  \geq\, c_\eps\, \|u\|_{L^2(\Omega)}^2
$$
for some  $c_\eps>0$ and all $u\in H^1(\Omega)$. Hence 
$\lambda_1(\sigma, \Omega) >0$ for any $\sigma\in \Sigma_m$ with $m>0$. 
\end{proof}

\noindent The assumption $N\geq 2$ in Proposition \ref{prop-inf1}  is crucial, see Section \ref{ss-1dim} for related results in dimension one. 


\subsection{The supremum} The main object of our interest here is the quantity 
\begin{equation} \label{supinf}
\Lambda_1(m, \Omega) = \sup_{\sigma\in \Sigma_m} \lambda_1(\sigma, \Omega),
\end{equation} 
and the function $\sigma\in \Sigma_m$ which realises the above supremum. It will be showen that, contrary to $\inf_{\sigma\in \Sigma_m} \lambda_1(\sigma, \Omega)$, the supremum $\Lambda_1(m, \Omega)$ is achieved on $\Sigma_m$. 
We will give an explicit characterisation of the maximising $\sigma$, and prove sharp two-sided estimates for the related maximal eigenvalue in terms of $m$ and the volume of $\Omega$. The existence of the maximising $\sigma$ in \eqref{supinf} is related to the following simple observation:  

\medskip

\noindent If $\hat\sigma\in\Sigma_m$ is such that the corresponding minimiser $\hat u$ of $Q[\hat\sigma, \cdot \, ]$ is {\em constant on $\Gamma$}, then $\Lambda_1(m, \Omega) = \lambda_1(\hat\sigma, \Omega)$. Indeed, for any $\sigma\in \Sigma_m$ we then
have
\begin{equation} \label{suff}
\inf_{u \in H^1(\Omega)} Q[\sigma, u] \leq Q[\sigma, \hat u] = Q[\hat\sigma, \hat u] = \lambda_1(\hat\sigma, \Omega).
\end{equation}
It thus suffices to find a suitable candidate for $\hat\sigma$. To do so, we consider the corresponding limiting problem for $\sigma\to \infty$, which is associated with the Laplace operator $-\Delta_D^\Gamma$ in $L^2(\Omega)$ subject to Dirichlet boundary conditions on $\Gamma$ and to Neumann boundary conditions on $\partial\Omega\setminus\Gamma$. More precisely, $-\Delta_D^\Gamma$ is generated by the closed quadratic form 
\begin{equation} \label{dn-form}
\Q_\Gamma[u] = \int_\Omega |\nabla u(x)|^2\, dx, \qquad \D(\Q_\Gamma)= \{ u\in H^1(\Omega)\, : \, u |_\Gamma =0 \}.
\end{equation}
The compactness of the imbedding $\D(\Q_\Gamma) \hookrightarrow L^2(\Omega)$ implies that the spectrum of $-\Delta_D^\Gamma$ is purely discrete. Let  $E_j(\Gamma)$ be the non-decreasing sequence of its eigenvalues and let $\varphi_j$ be the associated normalised eigenfunctions. Hence
\begin{equation} \label{e1}
E_1(\Gamma) = \min_{u\in \D(\Q_\Gamma)} \frac{\Q_\Gamma[u]}{\|u\|_{L^2(\Omega)}^2}.
\end{equation}
If $\Gamma$ has a positive measure, then in view of the Poincar\'e inequality we have $E_1(\Gamma)>0$.
Recall that  the operator domain $\D(-\Delta_D^\Gamma) $ of $-\Delta_D^\Gamma$ satisfies
\begin{equation} \label{domain}
\D(-\Delta_D^\Gamma) \subseteq \big\{ u\in H^2(\Omega)\, : \ u |_\Gamma =0 \  \ \wedge\ \  \partial_n u |_{\partial\Omega\setminus\Gamma} =0 \quad \text{a. e. } \big\}.
\end{equation}
Now for $\xi \in (0,E_1(\Gamma))$ we define 
$$
U_\xi = (-\Delta_D^\Gamma-\xi)^{-1}\,\id,
$$
where $\id$ denotes the function identically equal to $1$ on $\Omega$. Since $(-\Delta_D^\Gamma-\xi)^{-1}$ is positivity preserving, we have $U_\xi >0$ in $\Omega$. Hence from the strong maximum principle and the fact that $(-\Delta_D^\Gamma-\xi)^{-1}$ maps  $L^2(\Omega)$ into $\D(-\Delta_D^\Gamma)$ it follows that
\begin{equation} \label{hopf}
\partial_n U_\xi   \big |_{\Gamma} < 0  \quad \text{and} \quad \partial_n U_\xi   \big |_{\partial\Omega\setminus\Gamma} =0 \qquad a.e.
\end{equation}
Together with $U_\xi$ we introduce the function $F:(0,E_1(\Gamma)) \to \R$ given by
\begin{equation} \label{eq-f}
F(\xi) := \xi^2 \int_\Omega U_\xi\, dx + \xi\, |\Omega|, \qquad \xi \in (0,E_1(\Gamma)),
\end{equation}
where $ |\Omega|$ denotes the volume of $\Omega$.  

\begin{lemma} \label{lem-f}
The function $F$ is a $C^2$ increasing convex bijection from $(0,E_1(\Gamma))$ onto $(0,\infty)$.
\end{lemma}

\begin{proof}
Let 
$$
g(\xi):= \int_\Omega U_\xi\, dx = (\id,\,  (-\Delta_D^\Gamma-\xi)^{-1} \id )_{L^2(\Omega)},
$$ 
so that $F(\xi) = \xi^2 g(\xi)+\xi\, |\Omega|$. From the resolvent identity 
\begin{equation} \label{resolvent}
(-\Delta_D^\Gamma-\xi)^{-1} -(-\Delta_D^\Gamma-\xi')^{-1} = (\xi-\xi')\, (-\Delta_D^\Gamma-\xi)^{-1}(-\Delta_D^\Gamma-\xi')^{-1}
\end{equation}
we easily find out that $g'(\xi) = \|U_\xi\|^2_{L^2(\Omega)}>0$. This shows that $F$ is increasing. Moreover, using \eqref{resolvent} again we get 
$$
g''(\xi) = 2\, \big(U_\xi,\,  (-\Delta_D^\Gamma-\xi)^{-1}\, U_\xi \big)_{L^2(\Omega)} >0,
$$
which implies that $F$ is convex. It remains to show that $F: (0,E_1(\Gamma)) \to(0,\infty)$ is surjective. Obviously, $F(\xi) \to 0$ as $\xi \to 0$. On the other hand, using the explicit expression for the integral kernel of $(-\Delta_D^\Gamma-\xi)^{-1}$ we obtain
\begin{equation} \label{diver}
g(\xi) = \frac{1}{E_1(\Gamma)-\xi}\, \Big(\int_\Omega \varphi_1\, dx\Big)^2 +\mathcal{O}(1), \qquad \xi \nearrow E_1(\Gamma).
\end{equation}
Since $F$ is continuous, it follows that $F$ maps $(0,E_1(\Gamma))$ onto $(0,\infty)$.  
\end{proof}

\medskip

\noindent The above Lemma allows us to introduce the function $\xi$ on $(0,\infty)$ given by 
\begin{equation} \label{implicit}
\xi (m)  = F^{-1}(m), \qquad  m >0. 
\end{equation}
In view of Lemma \ref{lem-f} we easily see that $\xi$ is concave and maps $(0,\infty)$ onto $(0,E_1(\Gamma))$. 

\begin{theorem} \label{thm-sup}
The supremum $\Lambda_1(m, \Omega)=\sup_{\sigma\in \Sigma_m} \lambda_1(\sigma, \Omega)$ is attained for any $m>0$ and satisfies  
\begin{equation} \label{max}
\Lambda_1(m, \Omega)=   \lambda_1(\sigma_m, \Omega) = \xi(m),  \quad \text{where \ \ } \sigma_m = -\xi(m)\, \partial_n U_{\xi(m)}\big |_{\partial\Omega}.
\end{equation}
Moreover, the maximiser $\sigma_m$ is unique in $\Sigma_m$.   
\end{theorem}

\begin{proof} 
As mentioned above, to prove that $\sigma_m$ is a maximiser it suffices to show that the minimiser of the functional $Q[\sigma_m, \cdot]$ is constant on $\Gamma$. 
Recall that $U_\xi\in H^2(\Omega)$. From the imbedding \eqref{trace-imbed} we find that $\sigma_m\in L^1(\partial\Omega)$. Moreover, since $-\Delta U_{\xi(m)} = \xi(m) U_{\xi(m)}+1$, the Green formula and equations \eqref{hopf}, \eqref{implicit} yield
\begin{align*}
\int_{\partial\Omega} \sigma_m\, d\nu & = -\int_{\partial\Omega} \partial_n U_{\xi(m)}\, d\nu = -\xi(m) \int_\Omega \Delta U_{\xi(m)}\, dx = F(\xi(m)) =m.
\end{align*}
This in combination with \eqref{hopf} shows that $\sigma_m\in\Sigma_m$.  Next we define
$$
u_m := \xi(m)\, U_{\xi(m)} + 1,
$$
so that
\begin{equation} \label{eq-ut}
-\Delta \, u_m = \xi(m)\, u_m \quad \text{in \ } \Omega, \qquad \partial_n u_m + \sigma_m =0 \quad \text{a.e. \ \ on \ } \partial\Omega.
\end{equation}
We claim that $u_m$ is a minimiser of $Q[\sigma_m, \cdot\, ]$. Indeed, by \eqref{eq-ut} we have $Q[\sigma_m, u_m] = \xi(m)$. Assume that $\lambda_1(\sigma_m, \Omega) < \xi(m)$. Then in view of Lemma \ref{lem-minimiser}  there exists a positive minimiser $\psi$ of $Q[\sigma_m, \cdot\, ]$ which satisfies equation \eqref{euler-lagr} with $\sigma=\sigma_m$. This in combination with \eqref{eq-ut} and integration by parts implies that $(\psi, u_m)_{L^2(\Omega)}=0$, which is in contradiction with the positivity of $\psi$ and $u_m$. We thus conclude that $\lambda_1(\sigma_m, \Omega) = \xi(m)$.

To show that  $\Lambda_1(m, \Omega)= \lambda_1(\sigma_m, \Omega)$ pick an arbitrary $\sigma\in\Sigma_m$. As already pointed out in \eqref{suff}, using \eqref{lambda} and the fact that $u_m = 1$ on $\Gamma$ we obtain
\begin{equation} \label{test-f}
\lambda_1(\sigma, \Omega)   \leq  Q[\sigma, u_m] = 
Q[\sigma_m, u_m] = \lambda_1(\sigma_m, \Omega).
\end{equation}
It remains to show the uniqueness of $\sigma_m$. To this end suppose that $\lambda_1(\bar\sigma, \Omega)=\lambda_1(\sigma_m, \Omega)$ for some $\bar\sigma\in\Sigma_m$. By the same argument used in \eqref{test-f} we find out that $u_m$ is a minimiser of $Q[\bar\sigma, \cdot\, ]$:
$$
\lambda_1(\sigma_m, \Omega)  =\lambda_1(\bar\sigma, \Omega) \leq  Q[\bar\sigma, u_m] = Q[\sigma_m, u_m] = \lambda_1(\sigma_m, \Omega).
$$
By Lemma \ref{lem-minimiser} it follows that $u_m$ satisfies the Euler-Lagrange equation \eqref{eq-ut} with $\sigma_m$ replaced by $\bar\sigma$. Hence $\bar\sigma= -\partial_n u_m  |_{\Gamma} = \sigma_m$ almost everywhere on $\Gamma$. 
\end{proof}

\begin{remark}
A slightly different optimising problem for two-dimensional domains was studied in \cite{cu}, where the authors addressed the question on which part of boundary one has to impose Dirichlet boundary conditions to minimise or maximise the lowest eigenvalue of a mixed Dirichlet-Neumann boundary value problem.
\end{remark}


\subsection{The case $N=1$} \label{ss-1dim}
In the case of dimension one we have $\Omega = (a,b)$ and 
$$
\Sigma_m = \{ \sigma = (\sigma(a), \sigma(b))\, : \ \sigma(a),\, \sigma(b) \geq 0 \ \wedge \ \sigma(a)+\sigma(b) =m \}.
$$
Since Theorem \ref{thm-sup} holds true in any dimension, the maximiser of $\lambda_1(\cdot, \Omega)$ is given by 
that $\sigma$ for which the associated minimiser $u$ in \eqref{lambda} satisfies $u(a)=u(b)$. In other words
$$
\sigma_m =\Big (\frac m2, \frac m2\Big).
$$ 

\smallskip

\noindent On the other hand, the claim of Proposition \ref{prop-inf1} fails if $N=1$ since the capacity of a point on one-dimensional bounded intervals is positive. Consequently, the functional $\lambda_1(\sigma, \Omega)$ admits minimisers on $\Sigma_m$ and the resulting minimum is positive for any $m>0$. 

\begin{proposition} \label{lem-1dim}
Let $\Omega=(a,b)$. Then for any $m>0$ it holds
\begin{equation} \label{eq-1dim}
\inf_{\sigma\in \Sigma_m} \lambda_1(\sigma, \Omega) = \lambda_1(\sigma_1, \Omega) =\lambda_1(\sigma_2, \Omega) \geq \frac 14 \, \big(b-a +\frac{1}{2m}\big)^{-2},
\end{equation}
where $\sigma_1 = (m,0)$ and $\sigma_2 = (0,m)$.
\end{proposition}

\begin{proof}
We start by proving that $\sigma_1$ and $\sigma_2$ are minimisers of $\lambda_1(\cdot, \Omega)$ on $\Sigma_m$. 
Given $u\in H^1(a,b)$ we set $\hat u(x) =  u(a+b-x)$. Let $\sigma\in\Sigma_m$ and denote by $u_\sigma$ the positive 
normalised minimiser of $Q[\sigma, \cdot\, ]$. It is easily seen that 
\begin{equation} \label{aux}
\sigma(a) \leq \sigma (b) \ \Rightarrow \ u_\sigma(a) \geq u_\sigma(b), \qquad \sigma(a) > \sigma (b) \ \Rightarrow \ u_\sigma(a) \leq u_\sigma(b)
\end{equation}
This follows from the fact that if $\sigma(a) \leq \sigma (b)$ and $u$ is such that $0\leq u(a) < u(b)$, then $Q[\sigma, u] >Q[\sigma, \hat u]$. The same argument proves the second implication in \eqref{aux}.  

Assume first that $\sigma(a) > \sigma (b)$. Then in view of \eqref{aux} and the fact that $\sigma(a)+\sigma(b)=m$ we get 
\begin{align*}
\lambda_1(\sigma_1, \Omega)  \leq Q[\sigma_1, u_\sigma] &= \lambda_1(\sigma, \Omega) + m \, u_\sigma^2(a) -\sigma(a)\, u_\sigma^2(a)-\sigma(b)\, u_\sigma^2(b)  \leq  \lambda_1(\sigma, \Omega).
\end{align*}
On the other hand, if $\sigma(a) \leq \sigma (b)$, then again with the help of \eqref{aux} it follows that 
\begin{align*}
\lambda_1(\sigma_1, \Omega)  \leq Q[\sigma_1, \hat u_\sigma] &= \lambda_1(\sigma, \Omega) + m \, u_\sigma^2(b) -\sigma(a)\, u_\sigma^2(a)-\sigma(b)\, u_\sigma^2(b)  \leq  \lambda_1(\sigma, \Omega).
\end{align*}
Hence $\sigma_1$ is a minimiser of $\lambda_1(\cdot, \Omega)$. The proof for $\sigma_2$ is completely analogous. Obviously,  $\lambda_1(\sigma_1, \Omega)= \lambda_1(\sigma_2, \Omega)$.  To prove the inequality in \eqref{eq-1dim} we note that for any $u\in H^1(a,b)$ it holds
\begin{align*}
\int_a^b \Big(u'(x) -\frac{u(x)}{2(x-a+\frac{1}{2m})} \Big)^2\, dx & = \int_a^b |u'(x)|^2\, dx +m\, u^2(a) - \frac{u^2(b)}{2(b-a+\frac{1} {2m})} \\
& \quad - \frac 14\, \int_a^b \frac{u^2(x)}{(x-a+\frac{1}{2m})^2} \, dx,
\end{align*}
where we have integrated by parts to evaluate the mixed term. It follows that 
$$
\int_a^b |u'(x)|^2\, dx +m\, u^2(a)  \, \geq \frac 14\, \int_a^b \frac{u^2(x)}{(x-a+\frac{1}{2m})^2} \, dx \qquad \forall\ u\in H^1(a,b),
$$
which yields the sought lower bound in \eqref{eq-1dim}.
\end{proof}

\begin{remark}
It is clear from the proof of Proposition \ref{lem-1dim} that $\sigma_1$ and $\sigma_2$ are the only minimisers of $\lambda_1(\cdot, \Omega)$. Note also that the explicit form of $\sigma_1$ and $\sigma_2$ is reminiscent of the properties os the sequence $\sigma_n$ used in the proof of Proposition \ref{prop-inf1}. 
\end{remark}


\subsection{\bf The maximal eigenvalue $\Lambda_1(m, \Omega)$}

Theorem \ref{thm-sup} gives us information about the asymptotic behaviour of $\Lambda_1(m, \Omega)$ for $m\to 0$ as well as for $m\to \infty$. Indeed by \eqref{diver} 
\begin{equation} \label{m-infty}
\lim_{m\to\infty} \Lambda_1(m, \Omega) = E_1(\Gamma).
\end{equation}

\noindent Moreover, by Lemma \ref{lem-f} it follows that $\Lambda_1(m,\Omega)$ is a concave increasing function of $m$.  As for the behaviour of $\Lambda_1(m)$ for small values of $m$, by using a test function equal to a constant we see that $\Lambda_1(m, \Omega)\to 0$ when $m\to 0$. Moreover, by the resolvent equation \eqref{resolvent} and \eqref{eq-f} we get
$$
F(\xi) = \xi \, |\Omega| + \xi^2 \int_\Omega U_0(x)\, dx + o(\xi^2) \qquad \xi\to 0.
$$
In view of \eqref{implicit} and \eqref{max} we then get
\begin{equation} \label{m-zero}
\Lambda_1(m, \Omega) = m\  |\Omega|^{-1} + o(m), \qquad m\to 0+.
\end{equation}

\smallskip

\noindent A natural question is how to estimate $\Lambda_1(m, \Omega)$ for a fixed value of $m$.  It turns out that to this end it is not convenient to use directly the equation for $\Lambda_1(m, \Omega)$ given by Theorem \ref{thm-sup}, because we have very little information about the function $U_\xi$. Instead, we are going to employ merely
the fact that the corresponding minimiser is constant on $\Gamma$.

\begin{proposition} \label{prop-ub}
For any $m>0$ it holds
\begin{equation} \label{eq-lowerb}
\Lambda_1(m, \Omega) \, \geq \, \frac{m\, E_1(\Gamma)}{m+ |\Omega|\, E_1(\Gamma)} .
\end{equation}
\end{proposition}

\begin{proof}
Since $\Lambda_1(m, \Omega)=\lambda_1(\sigma_m, \Omega)$ admits a normalised eigenfunction which is constant on $\Gamma$, by Theorem \ref{thm-sup}, we have 
$$
\Lambda_1(m, \Omega) = \inf_{u\in \F} Q[\sigma_m, u],
$$
where 
\begin{equation}  \label{subspace}
\F = \big\{ u\in H^1(\Omega):\ \|u\|_{L^2(\Omega)}=1, \ \exists \, k \geq 0\, : \, u\big |_\Gamma = k\big\}.
\end{equation}
Now let $u\in\F$ and let $k$ be the corresponding constant in \eqref{subspace}. Then
\begin{align*}
Q[\sigma_m, u] & = \int_\Omega |\nabla u|^2\, dx + m\, k^2 = \int_\Omega |\nabla (u-k)|^2\, dx + m\, k^2 \\
& \geq E_1(\Gamma)  \int_\Omega |u-k|^2\, dx + m\, k^2 = E_1(\Gamma) (1- 2 k \int_\Omega u\, dx + k^2\, |\Omega| ) +m\, k^2,
\end{align*}
where we have used the fact that the function $u-k$ belongs to the form domain $\D(\Q_\Gamma)$ of the operator $-\Delta_D^\Gamma$, see \eqref{dn-form}. By the Cauchy-Schwarz inequality we have $|\int_\Omega u\, dx| \leq \sqrt{|\Omega|}$. Consequently 
$$
Q[\sigma_m, u] \, \geq\,  E_1(\Gamma) (1-  k \sqrt{|\Omega|})^2 +m\, k^2.
$$
Minimising the right hand side with respect to $k$ then gives 
$$
Q[\sigma_m, u] \, \geq \, \frac{m\, E_1(\Gamma)}{m+ |\Omega|\, E_1(\Gamma)} \qquad \forall\ u\in\F.
$$
This yields \eqref{eq-lowerb}. 
\end{proof}

\noindent In order to estimate $\Lambda_1(m, \Omega)$ from above by a quantity comparable with the lower bound \eqref{eq-lowerb} we employ a test function which results from an ''interpolation'' between a constant function and the eigenfunction $\varphi_1$ of $-\Delta_D^\Gamma$ relative to $E_1(\Gamma)$. 

\begin{proposition}  \label{prop-lb}
For any $m>0$ it holds
\begin{equation} \label{eq-upperb}
\Lambda_1(m, \Omega) \, \leq  \frac{2\, m\,
E_1(\Gamma)}{m+|\Omega|\, E_1(\Gamma) 
+\sqrt{(|\Omega|\, E_1(\Gamma) -m)^2 +4\, \gamma_1^2\, m\, E_1(\Gamma)}}\, ,
\end{equation}
where $\gamma_1 = \int_\Omega \varphi_1$. 
\end{proposition}

\begin{proof}
Let $\sigma\in \Sigma_m$. We consider a family of test functions given by 
\begin{equation} \label{test-function}
f_t(x) = |\Omega|\, \, (1-t)\, \varphi_1(x) + \gamma_1\, t, \qquad t\geq 0.
\end{equation}
Then $f_t\in H^1(\Omega)$ for all
$t\geq 0$.
A direct calculation shows that 
$$
Q[\sigma, f_t] =
\frac{E_1(\Gamma)\,  |\Omega|\, \gamma_1^{-2} (1-t)^2 +m\, |\Omega|^{-1}\, t^2}{1+(|\Omega|\, \gamma_1^{-2}
-1)(1-t)^2}
$$ 
attains its minimum at
$$
t_0 = \frac{E_1(\Gamma) \,|\Omega|\, +m\,
-\sqrt{(|\Omega|\, E_1(\Gamma)\, -m)^2 +4\,
\gamma_1^2\, m\, E_1(\Gamma)}}{2\,
(|\Omega|-\gamma_1^2)\, m\, |\Omega|^{-1}}.
$$
Since $t_0$ solves the equation
$$
|\Omega|\, E_1(\Gamma)\, (1-t)^2 = t\, m\,
|\Omega|^{-1}\,  \big(|\Omega|-t\, (|\Omega|-\gamma_1^2)\big)\, (1-t),
$$
we find out that
$$
\lambda_1(\sigma, \Omega) \leq Q[\sigma, f_{t_0}]= t_0\, m\,
|\Omega|^{-1} = \frac{2\, m\,
E_1(\Gamma)}{m+|\Omega|\, E_1(\Gamma) 
+\sqrt{(|\Omega|\, E_1(\Gamma) -m)^2 +4\, \gamma_1^2\, m\, E_1(\Gamma)}}.
$$
\end{proof}

\begin{remark}
The right hand side of \eqref{eq-upperb} is obviously larger than the right hand side of \eqref{eq-lowerb} since 
$$
\gamma_1^2 =  |\Omega| -\frac 12\, \int_\Omega\int_\Omega |\varphi_1(x)-\varphi_1(x')|^2\, dx\, dx'.
$$
Hence the upper bound \eqref{eq-upperb} coincides with the lower bound \eqref{eq-lowerb} if and only if $\varphi_1$ is constant, in other words if and only if $\Gamma=\emptyset$, in which case we have $E_1(\Gamma)=\Lambda_1(m, \Omega)=0$.
\end{remark}

\begin{remark}
Note that in view of \eqref{m-infty} and \eqref{m-zero} the estimates \eqref{eq-lowerb}, \eqref{eq-upperb} are sharp in the limit $m\to \infty$ as well as in the limit $m\to 0$.
\end{remark}

\begin{corollary} \label{cor-lambdam}
We have
$$
 \frac{m\, E_1(\Gamma)}{m+ |\Omega|\, E_1(\Gamma)} \, \leq \, \Lambda_1(m, \Omega) \, \leq \,  \frac{2\, m\, E_1(\Gamma)}{m+ |\Omega|\, E_1(\Gamma)} \qquad \forall\ m>0.
$$
\end{corollary}

\begin{proof}
This follows immediately from Propositions \ref{prop-ub} and \ref{prop-lb}.
\end{proof}


\section{\bf Estimates on $\lambda_1(\sigma, \Omega)$}  \label{s-convex}

In this section we are going to study the properties of $\lambda_1(\sigma, \Omega)$ for a fixed $\sigma$.  
This problem has attracted a considerable attention mainly in the case when $\sigma$ is constant. An extension of the Faber-Krahn inequality, well-known for the Dirichlet-Laplacian, was established first in \cite{boss1} in dimension two and later in \cite{da2} for any dimension, see also \cite{bg}. Monotonicity properties of $\lambda_1(\sigma, \Omega)$ with respect to the domain shrinking were studied in \cite{pw, gs1}.  Various bounds on $\lambda_1(\sigma, \Omega)$ in terms of eigenvalues of Dirichlet and (or) Neumann Laplacian were found in \cite{ph, sp, sp2}.

\smallskip

Our aim is to estimate $\lambda_1(\sigma,\Omega)$ only in terms of $\sigma$ and the geometric properties of $\Omega$.  For this purpose we introduce some notation. Let 
$$
\delta(x) = \min_{y\in\partial\Omega} |x-y|, \qquad x\in\Omega
$$
be the distance between a point $x$ and the boundary of $\Omega$, and let 
$$
R_\Omega = \sup_{x\in\Omega}\,  \delta(x)
$$
be the inradius of $\Omega$. Finally, let $K_N$ denote the lowest eigenvalue of the
Dirichlet Laplacian on a unit ball in $\R^N$. It is known that the lowest eigenvalue $\lambda_1^D(\Omega)$ of the Dirichlet-Laplacian on a convex domain $\Omega$ can be estimated in terms of the inradius as follows:
\begin{equation} \label{two-sided-dir}
\frac 14\  R_\Omega^{-2}\, \leq \, \lambda_1^D(\Omega) \, \leq \, K_N \ R_\Omega^{-2}.
\end{equation}
Here the upper bound follows by scaling and monotonicity of $\lambda_1^D(\Omega)$ with respect to the domain enlarging, while
the lower bound is a consequence of the Hardy inequality for Dirichlet-Laplacians on convex domains
\begin{equation} \label{eq:hardy-dir}
\int_{\Omega} |\nabla u|^2\, dx \,
\geq \, \frac 14\, \int_\Omega\, \frac{|u|^2}{\delta^2}\, \, dx \qquad \forall\ u \in H^1_0(\Omega),
\end{equation}
see e.g. \cite[Sect.5.3]{davies}. It is well-known that the constant $1/4$ on the right hand side of \eqref{eq:hardy-dir} is sharp
In order to get an idea how \eqref{two-sided-dir} should be modified when $\lambda_1^D(\Omega)$ is replaced by $\lambda_1(\sigma,\Omega)$ we will first study the scaling properties of the latter.

\begin{theorem} \label{prop:scaling}
Assume that $\sigma\in L^\infty(\partial\Omega)$ is non-negative. 
Let $\eps>0$ and let $\lambda_1(\sigma_\eps, \eps \Omega)$ be the lowest
eigenvalue of the Robin Laplacian on the rescaled domain $\eps\, 
\Omega$ with $\sigma_\eps(s) = \sigma(s/\eps)$. Then
\begin{align} \label{limit:eps}
\lim_{\eps\to 0}\, \eps\, \lambda_1(\sigma_\eps, \eps \Omega) & = 
 |\Omega|^{-1}  \textstyle\int_{\partial\Omega} \sigma\, d\nu .
 \end{align}
Moreover, 
\begin{align}
\lim_{\eps\to \infty}\, \eps^2\, \lambda_1(\sigma_\eps, \eps \Omega) & = 
E_1(\Gamma),  \label{eps:infty}
\end{align}
where $\Gamma=$ {\rm supp}\, $\sigma$ and 
$E_1(\Gamma)$ is given by \eqref{e1}. 
\end{theorem}

\begin{proof}
By a change of variables we obtain
\begin{equation} \label{scaling}
\lambda_1(\sigma_\eps, \eps \Omega) = \inf_{u\in H^1(\Omega)}  \frac{
\eps^{N-2}\, \int_{\Omega} |\nabla u|^2\, dx +\eps^{N-1}\,
\int_{\partial\Omega} \sigma\, |u|^2\, d\nu }{\eps^N\, \int_\Omega |u|^2\, dx}\, .
\end{equation}
Lemma \ref{lem-minimiser} implies that there exists a sequence
of positive minimisers $u_\eps\in H^1(\Omega)$ of problem \eqref{scaling}. We may suppose that  $\|u_\eps\|_{L^2(\Omega)}=1$ for all $\eps>0$. Hence
 \begin{equation} \label{eq-u-eps}
\lambda_1(\sigma_\eps, \eps \Omega) = 
\eps^{-2} \int_{\Omega} |\nabla u_\eps|^2\, dx +\eps^{-1}
\int_{\partial\Omega} \sigma\, |u_\eps|^2\, d\nu .
\end{equation}

\smallskip

\noindent Consider first the limit $\eps\to 0$. A simple test function argument with a constant function shows that
\begin{equation} \label{1-upperb}
 \eps \, \lambda_1(\sigma_\eps, \eps \Omega)  \leq
\frac{\textstyle\int_{\partial\Omega} \sigma\, d\nu}{|\Omega|} \qquad \forall\ \eps >0.
\end{equation}
In view of \eqref{1-upperb} and
\eqref{eq-u-eps} 
$$
 \eps \, \lambda_1(\sigma_\eps, \eps \Omega)= \eps^{-1}\, \int_{\Omega} |\nabla u_\eps|^2\, dx
 + \int_{\partial\Omega} \sigma\, |u_\eps|^2\, d\nu\, \leq \, \frac{\textstyle\int_{\partial\Omega} \sigma\, d\nu}{|\Omega|}\, .
$$
We thus have $\|\nabla u_\eps\|_{L^2(\Omega)}\to 0$ as $\eps\to 0$. Let $v_\eps$ be a subsequence of $u_\eps$. 
Since $v_\eps$ is bounded in $H^1(\Omega)$, it contains another
subsequence (which we still denote by $v_\eps$), such that $v_\eps$
converges weakly in $H^1(\Omega)$ to some $v$. Hence $\|v\|_{L^2(\Omega)}^2=1$. 
Moreover, the weak lower semicontinuity of $\int_{\Omega} |\nabla u|^2$ implies that  $\|\nabla v\|_{L^2(\Omega)}=0$ and therefore $
v= 1/\sqrt{|\Omega|}$ almost everywhere in $\Omega$.  We thus conclude that $v_\eps\to v$  
in $H^1(\Omega)$. Since this holds for any subsequence of $u_\eps$, we conclude that $u_\eps \to 1/\sqrt{|\Omega|}\, $ 
in $H^1(\Omega)$. By \eqref{trace-imbed} it follows that 
$$
\lim_{\eps\to 0} \|u-u_\eps\|_{L^2(\partial\Omega)}  =0.
$$
Since $\sigma\in L^\infty(\partial\Omega)$, in view of equation \eqref{eq-u-eps} we then have
$$
\liminf_{\eps\to 0}\, \eps \, \lambda_1(\sigma_\eps, \eps \Omega) \geq
\liminf_{\eps\to 0}\, \int_{\partial\Omega} \sigma\, |u_\eps|^2\, d\nu =
\int_{\partial\Omega} \sigma\, |u|^2\, d\nu = \frac{\textstyle\int_{\partial\Omega} \sigma\, d\nu}{|\Omega|}\, .
$$
This in combination with \eqref{1-upperb} proves \eqref{limit:eps}. To prove \eqref{eps:infty} we first note that 
\begin{equation} \label{upperb-infty}
\eps^2\, \lambda_1(\sigma_\eps, \eps \Omega) \, \leq\, E_1(\Gamma) \qquad \forall\ \eps>0,
\end{equation}
which follows by choosing the first eigenfunction $\varphi_1$ of the operator $-\Delta_D^\Gamma$ in $L^2(\Omega)$ as a test function in \eqref{scaling}. Now let $w_\eps$ be a subsequence of $u_\eps$. 
The sequence $w_\eps$ is then bounded in $H^1(\Omega)$ as $\eps\to\infty$, see \eqref{eq-u-eps}. Let $w$ be a weak limit of $w_\eps$ (or a suitable subsequence which we still denote by $w_\eps$) in $H^1(\Omega)$. Thus $\|w\|_{L^2(\Omega)}=1$.  From \eqref{eq-u-eps} and \eqref{upperb-infty} we conclude that 
$\int_{\partial\Omega} \sigma\, |w_\eps|^2\, d\nu \to 0$ as $\eps\to\infty$. Since $w_\eps \to w$ strongly in $L^2(\partial\Omega)$, see \eqref {trace-imbed}, it follows that $\int_{\partial\Omega} \sigma\, |w|^2\, d\nu=0$. Consequently, $w(s) =0$ for almost every $s\in \Gamma$ which implies that $w$ belongs to the form domain D$(\Q_\Gamma)$ of the operator $-\Delta_D^\Gamma$, see \eqref{dn-form}.  By 
 the weak lower semicontinuity of $\int_{\Omega} |\nabla u|^2$ and \eqref{e1} we thus conclude that 
$$
\liminf_{\eps\to\infty} \big(\int_{\Omega} |\nabla w_\eps|^2\, dx +\eps \int_{\partial\Omega} \sigma\, |w_\eps|^2\, d\nu \big) \, \geq\, \int_{\Omega} |\nabla w|^2\, dx \geq E_1(\Gamma). 
$$
On the other hand, from \eqref{eq-u-eps} and \eqref{upperb-infty} we get
$$
\limsup_{\eps\to\infty} \big(\int_{\Omega} |\nabla w_\eps|^2\, dx +\eps \int_{\partial\Omega} \sigma\, |w_\eps|^2\, d\nu \big) \,  \leq E_1(\Gamma). 
$$
Hence $w=\varphi_1$. Since $w_\eps$ was arbitrary, we conclude that $u_\eps\to \varphi_1$ weakly in $H^1(\Omega)$, which implies
$$
\liminf_{\eps\to\infty} \eps^2\, \lambda_1(\sigma_\eps, \eps \Omega)  \, \geq\, \int_{\Omega} |\nabla\varphi_1|^2\, dx =E_1(\Gamma).
$$
In view of \eqref{upperb-infty} this yields \eqref{eps:infty}.  
\end{proof}

\begin{remark}
The asymptotic behaviour \eqref{limit:eps} appears only when we deal with the first eigenvalue $\lambda_1(\sigma, \Omega)$. In fact, for any $\sigma \geq 0$ we have by the variational principle $\lambda_j^N(\Omega) \leq \lambda_j(\sigma, \Omega) \leq \lambda_j^D(\Omega)$, where $\lambda_j^N(\Omega),  \lambda_j^D(\Omega)$ and $\lambda_j(\sigma, \Omega)$ denote the $j$th eigenvalues of the Neumann, Dirichlet and Robin Laplacian respectively. By scaling 
$$
 \eps^{-2}\, \lambda_j^N(\Omega) = \lambda_j^N(\eps\Omega)\, \leq\, \lambda_j(\sigma_\eps, \eps\Omega)\, \leq\,  \lambda_j^D(\eps\Omega)  =  \eps^{-2}\, \lambda_j^D(\Omega). 
$$
Since $\lambda_j^N(\Omega) >0$ whenever $j\geq 2$, it follows that $\lambda_j(\sigma_\eps, \eps\Omega) \asymp \eps^{-2}$ for all
$j\geq 2$. This shows that the Robin Laplacian differs from both Dirichlet and Neumann Laplacians in the sense that its lowest eigenvalue scales, when $\eps\to 0$, in a different way than all the other eigenvalues. 
\end{remark}

\noindent Theorem \ref{prop:scaling} says that $\lambda_1(\sigma,\Omega) \sim R^{-1}_\Omega$ as $R_\Omega \to 0$ and therefore inequality \eqref{two-sided-dir} must fail if we replace $\lambda_1^D(\Omega)$ by $\lambda_1(\sigma,\Omega)$. 

\noindent We are going to prove an analogue of \eqref{two-sided-dir} for $\lambda_1(\sigma,\Omega)$ in the case when $\sigma$ is {\it constant} and $\Omega$ is {\it convex}. For an upper bound we will use the results of the previous section. In order to find an appropriate lower bound we start by proving a modified version of Hardy inequality \eqref{eq:hardy-dir}.

\begin{lemma} \label{hardy-convex}
Let $\sigma\geq 0$ and assume that $\Omega$ is convex. Then the inequality 
\begin{equation} \label{eq:hardy-convex}
\int_{\Omega} |\nabla u(x)|^2\, dx + \sigma\int_{\partial\Omega} |u(s)|^2\, d\nu(s) \,
\geq \, \alpha\sigma(1-\alpha\sigma)\, \int_\Omega\,
\frac{|u(x)|^2}{(\delta(x)+\alpha)^2}\, \, dx.
\end{equation}
holds true for all $u\in H^1(\Omega)$ and any $\alpha>0$.
\end{lemma}

\begin{proof}
The inequality is obvious for $\sigma=0$. Hence we may assume that $\sigma>0$. In view of the regularity 
of $\Omega$ is suffices to prove \eqref{eq:hardy-convex} for all $u\in C^1(\overline{\Omega})$. 
Since $|\nabla\delta|=1$ almost everywhere, we have
\begin{align} \label{integr}
\int_\Omega\, | \nabla u -\frac{\alpha\, \sigma\,
u}{\delta+\alpha}\, \nabla\delta\, |^2\, dx &= \int_\Omega\Big(
|\nabla u|^2 +\frac{\alpha^2\sigma^2\, u^2}{(\delta+\alpha)^2}\,
-\frac{2\, \alpha\, \sigma\, u}{\delta+\alpha}\, \nabla
u\cdot\nabla\delta \Big)\, dx.
\end{align}
Moreover, integration by parts gives
\begin{align} \label{perpartes}
2\int_\Omega \frac{u}{\delta+\alpha}\, \nabla u\cdot\nabla\delta\,
dx &= \int_{\partial\Omega}\,
\partial_n\delta(s)\, \frac{u^2(s)}{\alpha}\, d\nu(s) - \int_\Omega\,
\frac{u^2\, \Delta\delta}{\delta+\alpha}\, dx + \int_\Omega\,
\frac{u^2}{(\delta+\alpha)^2}\, dx.
\end{align}
Recall that $|\partial_n\delta(s)| =1$.
Moreover, from the convexity of $\Omega$ follows that $\delta$ is
concave and therefore $\Delta\delta \leq 0$ in the sense of distributions. 
Hence inserting \eqref{perpartes} into \eqref{integr} we get
\begin{align*}
\int_{\Omega} |\nabla u|^2\, dx + \sigma\int_{\partial\Omega} u^2(s)\, d\nu(s) \,
& \geq \alpha\sigma(1-\alpha\sigma)\, \int_\Omega\,
\frac{u^2}{(\delta(x)+\alpha)^2}\, \, dx + \int_\Omega\, | \nabla u
-\frac{\alpha\sigma\, u}{\delta+\alpha}\, \nabla\delta |^2\, dx,
\end{align*}
which proves the statement.
\end{proof}

\medskip

\noindent Armed with Lemma \ref{hardy-convex} we can state the following

\begin{theorem} \label{thm-convex}
Assume that $\Omega$ is convex and that $\sigma>0$ is constant. Then
\begin{equation} \label{two-sided-robin}
\frac 14\, \frac{\sigma}{ R_\Omega(1+\sigma\, R_\Omega)}\, \leq\,  \lambda_1(\sigma,\Omega) \,  \leq \, 2\, K_N\,  \frac{ \sigma}{
R_\Omega(1+\sigma\, R_\Omega)}\, .
\end{equation}
\end{theorem}

\smallskip

\begin{remark}
The expression 
\begin{equation} \label{eqv}
\frac{\sigma}{ R_\Omega(1+\sigma\, R_\Omega)}
\end{equation} 
which appears on both sides of inequality \eqref{two-sided-robin} is proportional to $R^{-1}_\Omega$ for $R_\Omega \to 0$ and to $R_\Omega^{-2}$ for $R_\Omega\to \infty$. This is in agreement with Theorem \ref{prop:scaling}. It is also worth noticing that \eqref{eqv} is, just like $\lambda_1(\sigma, \Omega)$, an increasing function of $\sigma$, and that in the limit $\sigma\to \infty$ the two-sided inequality \eqref{two-sided-robin} turns, up to the multiplicative factor $2$, into \eqref{two-sided-dir} . 
\end{remark}

\begin{proof}[Proof of Theorem \ref{thm-convex}]
By inequality \eqref{eq:hardy-convex} we have
$$
\lambda_1(\sigma,\Omega) \geq \frac{\alpha \sigma(1-\alpha
\sigma)}{(R_\Omega+\alpha)^2}\, \qquad \forall\, \alpha>0.
$$
The lower bound in  \eqref{two-sided-robin} then follows by maximising the right hand side of the above inequality
with respect to $\alpha$.
As for the upper bound, we apply Theorem \ref{thm-sup} and Proposition \ref{prop-lb}  with $\Gamma=\partial\Omega$ to obtain  
\begin{equation} \label{aux-2}
\lambda_1(\sigma,\Omega) \, \leq \, \Lambda_1(\sigma\, |\partial\Omega|, \, \Omega) \, \leq \, 2\, \sigma
\left(\frac{|\Omega|}{|\partial\Omega|}
+\frac{\sigma}{\lambda_1^D(\Omega)}\right)^{-1}.
\end{equation}
Let us fix a system of coordinates in such a way that the ball $B(o,R_\Omega)$ centred in the origin $o$ satisfies $B(o, R_\Omega)\subseteq\Omega$. As mentioned above, the function $\delta(x)$ is concave on $\Omega$. Hence 
$$
\nabla\delta(x) \cdot (y-x) \geq \delta(y) -\delta(x)
$$
for all $x,y\in\Omega$ for which $\nabla\delta(x)$ exists. Since $\nabla\delta(s) = -n(s)$, we can insert $x=s\in\partial\Omega$ and $y=o$ in the above inequality to find out that $s \cdot n(s) \geq R_\Omega$ almost everywhere on $\partial\Omega$. Consequently, by the Gauss Theorem
$$
|\Omega| = \frac 1N\, \int_\Omega \text{div}\, x\ dx = \frac 1N\,
\int_{\partial\Omega} s \cdot n(s)\, d\nu(s) \, \geq \, |\partial\Omega|\, \frac{R_\Omega}{N}.
$$
This in combination with \eqref{two-sided-dir} and \eqref{aux-2} gives
$$
\lambda_1(\sigma,\Omega) \, \leq \, C_N\,  \frac{ \sigma}{ R_\Omega(1+\sigma\, R_\Omega)}\, , \qquad C_N= 2 \max\{N, K_N\}.
$$
Moreover, from the Li-Yau inequality, see \cite{ly} or \cite[p.305]{ll}, it follows that
$$
K_N \geq \, \frac{4\, N}{N+2}\, \Gamma\Big(1+\frac N2\Big)^{\frac 4N},
$$
where $\Gamma(\cdot)$ is the Euler gamma functions. By induction we then find out that $K_N \geq N$ for all $N\in\N$, 
which shows that $C_N=2\, K_N$.
This completes the proof of the upper bound in \eqref{two-sided-robin}. 
\end{proof}

\begin{remark}
By setting $\alpha= 1/2\sigma$ in \eqref{eq:hardy-convex} we obtain 
\begin{equation} \label{hardy-14}
\int_{\Omega} |\nabla u(x)|^2\, dx + \sigma\int_{\partial\Omega} |u(s)|^2\, d\nu(s) \,
\geq \, \frac 14\, \int_\Omega\,
\frac{|u(x)|^2}{(\delta(x)+\frac{1}{2\sigma})^2}\, \, dx,
\end{equation}
which is a special case of \cite[Thm3.1]{kl}, where a Hardy inequality for Robin-Laplacians with general (not necessarily constant) 
$\sigma$ was established.
However, inequality \eqref{hardy-14} would not allow us to arrive at the desired lower bound on 
$\lambda_1(\sigma, \Omega)$. For this reason we need the family of inequalities \eqref{eq:hardy-convex} parametrized by $\alpha$.
\end{remark}

\medskip



\section{\bf Acknowledgements} 
I thank Enrico Serra and Paolo Tilli for numerous helpful discussions. 
The support from the MIUR-PRINÕ08 grant for the project  ''Trasporto ottimo di massa, disuguaglianze geometriche e funzionali e applicazioni'' is gratefully acknowledged. 



\begin{thebibliography}{9999}

%
\bibitem[Ad]{ad} R. Adams, {\em Sobolev Spaces}, Elsevier Science Ltd,
Oxford , UK, 2003.
%
\bibitem[Bo1]{boss1}  M.H. Bossel: Membranes \'elastiquement li\'ees inhomog\`enes ou sur une surface: Une nouvelle extension du th\'eor\`eme isop\'erim\'etrique de Rayleigh-Faber-Krahn et de lÕin\'egalit\'e de Cheeger, 
, {\em C. R. Acad. Sci. Paris Ser. I Math.} {\bf 302} (1986), 47-50.
%
\bibitem[Bo2]{boss2}  M.H. Bossel: Membranes \'elastiquement li\'ees inhomog\`enes ou sur une surface: Une nouvelle extension du th\'eor\`eme isop\'erim\'etrique de Rayleigh-Faber-Krahn, {\em Z. Angew. Math. Phys.} {\bf 39} (1988), 733-742.
%
\bibitem[BG]{bg} D. Bucur, A. Giacomini: A variational approach to the isoperimetric inequality for the Robin eigenvalue problem, {\em Arch. Ration. Mech. Anal.} {\bf 198} (2010), 927--961.
%
\bibitem[CU]{cu} S.J. Cox,  P.X.Uhlig: Where best to hold a drum fast, {\em Siam Review} {\bf 45} (2003) 75-92.
%
\bibitem[Da1]{da} D. Daners: Robin boundary value problems on arbitrary domains,
{\em Trans. Amer. Math. Soc.} {\bf 352} (2000) 4207--4236.
%
\bibitem[Da2]{da2} D. Daners: A Faber-Krahn inequality for Robin problems in any space dimension, {\em Math. Ann.} 
{\bf 335} (2006) 767--785.
%
\bibitem[D]{davies} E.B. Davies: Spectral theory and differential operators, Cambridge University Press, UK, 1995.
%
\bibitem[GS]{gs1} T. Giorgi, R.G. Smits: Monotonicity results for the principal eigenvalue of the generalised Robin 
problem, {\em Illinois J. of Math.} {\bf 49} (2005) 1133-1143.
%
%
\bibitem[KL]{kl} H. Kova\v r\'{\i}k, A. Laptev: Hardy inequalities for Robin Laplacians,  {\em J. Funct. Anal.} 
{\bf 262} (2012) 4972--4985.
%
\bibitem[LP]{lp} M. Levitin, L. Parnovski: On the principal eigenvalue of a Robin problem with a large parameter, 
{\em Math. Nachr.} {\bf 281} (2008)   272-281.
%
\bibitem[LY]{ly} P.~Li and S.T.~Yau:  On the Schr\"odinger equation and the
eigenvalue problem, {\em Comm.~Math.~Phys.} {\bf 88} (1983) 309--318.
%
\bibitem[LL]{ll} E. H. Lieb, M. Loss: \textit{Analysis}.
Second edition. Graduate Studies in Mathematics \textbf{14}, American Mathematical Society, Providence, RI, 2001.
%
\bibitem[PW]{pw} L.E. Payne, H.F. Weinberger: Lower bounds for vibration frequencies of elastically sup-
ported membranes and plates. {\em J. Soc. Indust. Appl. Math.} {\bf 5}  (1957) 171--182.
%
\bibitem[Ph]{ph} G.A. Philippin: Some remarks on the elastically supported membrane, {\em Z. Angew. Math. Phys.} {\bf 29}
(1978), pp. 306Ð314.
%
\bibitem[Sp1]{sp} R.P. Sperb: Untere und obere Schranken f\"ur den tiefsten Eigenwert der elastisch
 gest\"uzten Membran, {\em Z. Angew. Math. Phys.} {\bf 23} (1972), 231--244.
 %
\bibitem[Sp2]{sp2} R.P. Sperb: Bounds for the first eigenvalue of the elastically supported membrane on convex
domains, {\em Z. Angew. Math. Phys.} {\bf 54} (2003), 879--903.
%
\end{thebibliography}
\end{document}